\documentclass[aps,superscriptaddress,11pt,twoside]{revtex4-2}

\usepackage{amsmath,latexsym,amssymb,verbatim,enumerate,graphicx,color,bm,mathtools,hyperref}
\usepackage{comment}
\usepackage{longtable}
\usepackage{graphicx}
\usepackage{float}
\usepackage{tikz}
\usetikzlibrary{shapes,positioning}
\usepackage{braket}

\usepackage{theorem}
\newtheorem{definition}{Definition}[section]

\newtheorem{lemma}[definition]{Lemma}

\newtheorem{theorem}[definition]{Theorem}

\newtheorem{oss}[definition]{Observation}

\def\squareforqed{\hbox{\rlap{$\sqcap$}$\sqcup$}}
\def\qed{\ifmmode\squareforqed\else{\unskip\nobreak\hfil
\penalty50\hskip1em\null\nobreak\hfil\squareforqed
\parfillskip=0pt\finalhyphendemerits=0\endgraf}\fi}
\def\endenv{\ifmmode\;\else{\unskip\nobreak\hfil
\penalty50\hskip1em\null\nobreak\hfil\;
\parfillskip=0pt\finalhyphendemerits=0\endgraf}\fi}
\newenvironment{proof}{\noindent \textbf{{Proof~}}}{\hfill\qed}
\newenvironment{remark}{\noindent \textbf{{Remark~}}}

\mathchardef\ordinarycolon\mathcode`\:
\mathcode`\:=\string"8000
\def\vcentcolon{\mathrel{\mathop\ordinarycolon}}
\begingroup \catcode`\:=\active
  \lowercase{\endgroup
  \let :\vcentcolon
  }

\makeatletter
\renewcommand*{\@fnsymbol}[1]{\ensuremath{\ifcase#1\or \ast\or \P\or \natural\or \flat\or \sharp\or \bullet\or \dagger\or \ddagger\or \|\or **\or \dagger\dagger\or \ddagger\ddagger \else\@ctrerr\fi}}
\makeatother

\def\cH{{\mathcal H}}
\def\cN{{\mathcal N}}

\newcommand{\be}{\begin{equation}}
\newcommand{\ee}{\end{equation}}
\newcommand{\ba}{\begin{align}}
\newcommand{\ea}{\end{align}}
\newcommand{\bea}{\begin{eqnarray}}
\newcommand{\eea}{\end{eqnarray}}

\begin{document}

\title{Profitable entanglement for channel discrimination}

\author{Samad Khabbazi Oskouei}
 \email{kh$\_$oskuei@yahoo.com}
 \affiliation{Department of Mathematics, Varamin-Pishva Branch
 Islamic Azad University, Varamin, 33817-7489 Iran}

\author{Stefano Mancini}
 \email{stefano.mancini@unicam.it}
 \affiliation{School of Science and Technology, University of Camerino,
              Via Madonna delle Carceri 9, I-62032 Camerino, Italy}
 \affiliation{INFN--Sezione Perugia, Via A. Pascoli, I-06123 Perugia, Italy}

\author{Milajiguli Rexiti}
 \email{milajiguli.milajiguli@gmail.com}
 \affiliation{School of Science and Technology, University of Camerino,
              Via Madonna delle Carceri 9, I-62032 Camerino, Italy}
 \affiliation{INFN--Sezione Perugia, Via A. Pascoli, I-06123 Perugia, Italy}


\begin{abstract}
We investigate the usefulness of side entanglement in discriminating between two generic qubit channels, {\  up to unitary pre- and post-processing,} and  
determine exact conditions under which it does enhance (as well as conditions under which it does not) the 
success probability.
This is done in a constructive way by first analyzing the problem for channels that are extremal in the set of completely positive and trace-preserving qubit linear maps and then for channels that are inside such a set. 
\end{abstract}
 
\maketitle



\section{Introduction}\label{sec:intro}

The relevance of discerning between two (or more) quantum channels stems from the fact that any physical process can be described as a quantum channel \cite{acin,Aneilsen,Sacchi2005,Mying}. It becomes even more pronounced when quantum channels are employed to encode information, like in quantum reading \cite{SPSM}.
{\  Broadly speaking,} the task of channel discrimination is challenging \cite{Jrenes,Brosgen}. In fact, albeit for a given input it can be traced back to the problem of state discrimination \cite{Hel}, it involves optimization of overall inputs. The performance is usually expressed in terms of minimum error probability.
Entanglement offers the possibility of enhancing such a performance by resorting to a strategy that uses a pair of systems such that only one undergoes the effect of the channel while both are measured (a strategy usually referred to as \emph{side entanglement})[20]. Remarkably this strategy can also enhance the distinguishability of entanglement-breaking channels \cite{MSacchi}. {\  It is known that an ancilla with the same
dimension as the input of the channels is always sufficient for optimal discrimination,
while this is not the case for the output dimension whenever
it is smaller \cite{Puz17}.}

However, until now, there is no clear and comprehensive picture of which channels can benefit from side entanglement, even at the level of the qubit.
Here we clarify this point by resorting to the characterization of qubit channels in terms of affine maps on 
$\mathbb{R}^3$ \cite{RSW}.

We find conditions under which side entanglement does enhance (as well as conditions under which it does not) the performance in discriminating between two generic qubit channels, {\  up to unitary pre- and post-processing}.
This is achieved in a constructive way by first analyzing the problem for channels that are extremal in the set 
$\mathfrak{N}$ of completely positive and trace-preserving qubit linear maps, {\  up to unitary pre- and post-processing}, and then for channels that are inside such a set.


\section{Preliminaries}\label{sec:pre}

A quantum channel is a linear, completely positive and trace preserving (CPTP) map on the set ${\cal D}(\cH)$ of density operators over a Hilbert space $\cH$. Every quantum channel ${\  {\cal M}}:{\cal D}(\cH)\to {\cal D}(\cH)$ can be expressed in the operator sum (Kraus) representation:
\be\label{kr}
{\  {\cal M}}(\rho)=\sum_j K_j\rho K_j^\dag,
\ee
where $K_j:\cH\to\cH$ are linear operators (called Kraus operators) satisfying the normalization condition $\sum_j K_j^\dag K_j=I$. The number of non-zero operators in the Kraus representation is called the Kraus rank.

In the following we will consider only qubit channels. For them, it is known that 
the maximum Kraus rank is four. 

In such a case, given a density operator $\rho=\frac{1}{2}\left(I + \boldsymbol{r} \cdot \boldsymbol{\sigma}\right)$ with $\boldsymbol{r}\in\mathbb{R}^3$,
$\|\boldsymbol{r}\|\leq 1$, 
and $\boldsymbol{\sigma}$ the vector of Pauli operators,
the channel ${\  {\cal M}}$ can always be represented as an affine transformation
in $\mathbb{R}^3$, that is
\be
\boldsymbol{r}\stackrel{{\  {\cal M}}}{\longrightarrow} 
\boldsymbol{M} \, \boldsymbol{r} + \boldsymbol{t},
\ee
where {\  $ \boldsymbol{M}$ is a $3\times 3$ real matrix and $ \boldsymbol{t}=(t_1,t_2,t_3)^T\in \mathbb{R}^3$.}

The matrix $\boldsymbol{M}$, via orthogonal transformations, can be brought to diagonal form 
${\  \boldsymbol{N}}={\rm diag}(\lambda_1,\lambda_2,\lambda_3)$.
{\ 
By virtue of that, any qubit channel ${\cal M}$ can be written as 
\begin{equation}
{\cal M}(\bullet)=V{\cal N}\left(U\bullet U^\dag\right) V^\dag,
\end{equation}
where ${\cal N}$ is a channel with associated diagonal matrix $ \boldsymbol{N}$ and vector $\boldsymbol{t}$, while $U$, $V$ are suitable (pre- and post-processing) quibit unitaries.
}

Then, the image through $\cN$ of the boundary of the Bloch sphere (pure states) is the ellipsoid described by the equation
\be
\left(\frac{x-t_1}{\lambda_1}\right)^2+ \left(\frac{y-t_2}{\lambda_2}\right)^2
+\left(\frac{z-t_3}{\lambda_3}\right)^2=1,
\ee
where $\lambda_k$s define the length of the axes of the ellipsoid and
$\boldsymbol{t}$ its center.

Note that not all ellipsoids within the Bloch sphere correspond to a quantum channel (completely
positive map).

{\  Now let us consider the set of qubit channels
\be
\mathfrak{N}:=\left\{ \cN:{\cal D}(\mathbb{C}^2)\to {\cal D}(\mathbb{C}^2)\,  | \, 
\cN \; \text{is linear and CPTP and} \; {\cal N} \leftrightarrow ( {\boldsymbol{N}}, {\boldsymbol{t}}) \right\}. 
\ee
}
When $\cN$ belongs to the closure of extreme points of $\mathfrak{N}$, it can be parametrized by 
{\  (see Sec. II.1 of Ref. \cite{RMP14} and reference therein)}
\be\label{tM}
\boldsymbol{t}=\begin{pmatrix}
0 \\ 0 \\ \sin(\phi-\theta)\sin(\phi+\theta)
\end{pmatrix},
\quad
{\  \boldsymbol{N}}=\begin{pmatrix}
\cos(\phi-\theta) & 0 & 0 \\ 0 & \cos(\phi+\theta) & 0 \\ 0 & 0 & \cos(\phi-\theta)\cos(\phi+\theta)
\end{pmatrix},
\ee
{\  with $\theta,\phi\in[0,\pi]$.} 
In term of Kraus operators we have
\be\label{kmatrix}
K_0=\begin{pmatrix}
\cos\theta && 0\\
0 && \cos\phi
\end{pmatrix}, \quad
K_1=\begin{pmatrix}
0 && \sin{ \phi} \\
\sin{\theta}  && 0
\end{pmatrix}.
\ee
When $\boldsymbol{t}=\boldsymbol{0}$ (unital channels) the extreme points of $\mathfrak{N}$ are known to be the maps that conjugate by a unitary matrix and in the $\lambda_k$ representation they correspond to four corners of a tetrahedron. The channels on the edges of such a tetrahedron correspond to ellipsoids that have exactly two points in common with the boundary of the Bloch
sphere.
These are called \emph{quasi-extreme points}
{\  because looking at extremality within $\mathfrak{N}$
they belong to the closure of extreme points of $\mathfrak{N}$,} but are not true extreme points.

When $\boldsymbol{t}\neq\boldsymbol{0}$ (non-unital channels) the extreme points of
$\mathfrak{N}$ correspond to maps for which a translation allows
the ellipsoid to touch the boundary of the Bloch sphere at two points.[21]

It is known \cite{RSW} that any quantum channel on $\mathfrak{N}$ can be written as convex combination
of two maps in the closure of extreme points.
In particular, a unital channel can be
written as a convex combination of two maps on the edges of the
tetrahedron.


\subsection{Channel Discrimination}

Let $\cN_{\omega_1}$ and $\cN_{\omega_2}$ be two qubit channels
characterized respectively by parameters $\omega_1,\omega_2\in\Omega$.[22]

Suppose each of the above channels appears with probability $\frac{1}{2}$,
and we want to discriminate between them by inputting a probe state $\rho\in{\cal D}(\mathbb{C}^2)$ and measuring the output.
The problem can be traced back to quantum states discrimination.
According to \cite{Hel}, given two quantum states
$\cN_{\omega_0} \left(\rho\right)$ and
$\cN_{\omega_1} \left(\rho\right)$
appearing with equal probability,
the maximum probability of success in discriminating between them reads
\be\label{sucpr}
P_{succ}{\left(\rho\right)}=\frac{1}{2}\left(1+\frac{1}{2}
\left\|{\  \Delta_{\cal N}(\rho)} \right\|_1
\right),
\ee
where
\be\label{dif1}
\Delta_{\cal N}(\rho):= \cN_{\omega_1} \left(\rho\right)- \cN_{\omega_2} \left(\rho\right),
\ee
and $\|\cdot\|_1$ denotes the trace-norm of an operator.

To achieve \eqref{sucpr}, the optimal measurement turns out to be the projection onto positive and negative eigenspaces of the operator $\Delta_{\cal N}(\rho)$.

Then, the problem of channel discrimination becomes finding the optimal input state $\rho$ which maximizes the trace distance of the output states.

This optimization can be restricted to the set of pure states ${\cal P}(\mathbb{C}^2)$ thanks to the convexity of the trace norm.


\subsection{Side Entanglement}

When a reference system is accessible for measurement, entanglement 
across reference and main systems 
can be exploited for the purpose of channel discrimination. 
We refer to this strategy as channel discrimination with {\it side entanglement}.

In this case we have to consider
\be
\Delta_{\text{id}\otimes {\cal N}}(\rho) := (\text{id}\otimes\cN_{\omega_1}) \left(\rho\right)
-(\text{id}\otimes \cN_{\omega_2})\left(\rho\right),
\ee
where $\rho\in{\cal P}(\mathbb{C}^2\otimes \mathbb{C}^2)$, assuming without loss of generality the reference system isomorphic to the main one.

According to Eq.\eqref{sucpr}, side entanglement turns out to be profitable if there exists a parameter's region (subset of 
$\Omega$) where the following inequality holds true:
\be\label{ineqnorm1}
\sup_{\rho\in{\cal P}(\mathbb{C}^2\otimes \mathbb{C}^2)}\|\Delta_{\text{id}\otimes {\cal N}}(\rho) \|_1
> \sup_{\rho\in{\cal P}(\mathbb{C}^2)} \|\Delta_{\cal N}(\rho)\|_1.
\ee

Since we will always deal with optimization on the set of pure states,
be it ${\cal P}(\mathbb{C}^2\otimes \mathbb{C}^2)$ or ${\cal P}(\mathbb{C}^2)$, 
in the following we will use interchangeably $|\rho\rangle\langle\rho|$ and $\rho$.

\bigskip

{\ 
In the following we shall consider the discrimination between two channels ${\cal N}_{\omega_1}$ and 
${\cal N}_{\omega_2}$ belonging to the set $\mathfrak{N}$.
This is equivalent to the discrimination between generic channels with associated matrices $\boldsymbol{M}_1$ and $\boldsymbol{M}_2$ that can be simultaneously diagonalized (or in other words two channels that have the same pre- and post-processing unitaries).
}


\section{Discriminating between extremal qubit channels}\label{sec:extreme}

In this Section we start considering the discrimination between two quasi-extremal qubit channels and then move on to the discrimination between two extremal qubit channels.

\begin{theorem}\label{theo1}
In the discrimination between two qubit channels that are quasi-extreme point maps of $\mathfrak{N}$, 
side entanglement is not useful.
\end{theorem}

\begin{proof}
According to the parametrization \eqref{tM} for quasi-extreme point maps
we must have {\  $\sin\theta=\sin\phi$. This can happen with $\cos\theta=\cos\phi$ or 
with $\cos\theta=-\cos\phi$. In the first case,}
by \eqref{kmatrix}, states along the $x$ direction of $\mathbb{R}^3$ are left invariant.  It follows that the optimal input state for channel discrimination 
(namely, the one that is mostly affected by  ${\cal N}_\omega$)
lies in the $y-z$ plane. Hence, it can be written as
$\ket{\psi}=\sqrt{r_0}\ket{+y}+\sqrt{r_1}\ket{-y}$,
where
$\ket{\pm y}$ denote the eigenvectors of Pauli operator $Y$.
As a consequence we have to optimize over $r_0,r_1$, such that $r_0+r_1=1$.

Now consider a two-qubit entangled state $\ket{\Psi}$. By the Schmidt decomposition
$\ket{\Psi}=\sum_{j=\pm}\sqrt{t_j}\ket{\tau_j}\ket{\tau_j}$, where $\{|\tau_j\rangle\}_j$ are orthonormal and 
$t_++t_-=1$. However, by local unitaries that do not affect the entanglement of $\ket{\Psi}$ we can rewrite it as $\ket{\Psi}=\sum_{j=\pm}\sqrt{\lambda_j}\ket{j y}\ket{j y}$. This means that we can express
$\Delta_{{\rm id}\otimes {\cal N}}(\Psi)$ as a
$2\times 2$ matrix restricting on the subspace spanned by $\{\ket{+y}\ket{+y}, \ket{-y}\ket{-y}\}$ and optimize over $\lambda_+,\lambda_-$ satisfying $\lambda_++\lambda_-=1$. However, this exactly coincides with the optimization over  $r_0,r_1$ of the state $\ket{\psi}$ without entanglement.

{\  In case $\cos\theta=-\cos\phi$, instead, by \eqref{kmatrix} we can realize that the states $\ket{\pm y}$ are flipped one into another. Hence they will be the most affected by the channel's action. Then, repeating the above reasoning for the entangled state $\ket{\Psi}$ we end up with the fact that $\|\Delta_{{\rm id}\otimes {\cal N}}(\Psi)\|_1\leq 
\| \Delta_{\cal N}(|\pm y\rangle\langle\pm y|)\|_1$.
}
\end{proof}

\bigskip

\begin{remark}
The channels of Theorem \ref{theo1} are on the edges of the tetrahedron and correspond to ellipsoids that have exactly two points in common with the boundary of the Bloch sphere.
Examples of this kind of channels are those having two Kraus operators,
one proportional to the identity and one proportional to one Pauli operator (bit-flip, phase-flip -- see also \cite{RMM}--, bit-phase-flip). Also those having two Kraus operators proportional to two distinct Pauli operators.
\end{remark}

\bigskip

Let us now address the discrimination of qubit channels that are extreme point in $\mathfrak{N}$.
We denote any such a channel as $\cN_{\phi,\theta}$, since it is characterized by Kraus operators \eqref{kmatrix}
with parameters $\phi,\theta$.

Consider the following states:
 \bea
 \ket{\psi}&=&\sum_{i=0}^1 a_i\ket i, \label{qb}\\
 \ket{\Psi}&=&\displaystyle\sum_{i, j=0}^{1} a_{ij}\ket{i j},
 \label{2qb}
 \eea
 with $a_i,a_{ij}\in\mathbb{C}$ such that $\sum_{i=0}^{1} \vert a_{i}\vert^2=1$ and $\sum_{i, j=0}^{1} \vert a_{ij}\vert^2=1$ respectively.

We want to study 
$\|\Delta_{{\cal N}_{\ {\phi, \theta}}}(\Psi)(\psi)\|_1$ vs
$\|\Delta_{\operatorname{id}\otimes{\cal N}_{\phi, \theta}}(\Psi)\|_1$ for the channel $\cN_{\phi,\theta}$.
Define
\begin{align}
\alpha &:=\cos^2\theta_1 -\cos^2\theta_2,
\label{re:18}\\
\beta &:=\cos^2\phi_1-\cos^2\phi_2,
\label{re:19}\\
\gamma_1 &:=\cos\phi_1\cos\theta_1-\cos\phi_2\cos\theta_2,
\label{re:20}\\
\gamma_2 &:=\sin\phi_1\sin\theta_1-\sin\phi_2\sin\theta_2,
\label{gamma2}
\end{align}
{together with 
\begin{equation}\label{gmM}
\gamma_m:=\left\{ 
\begin{matrix}
\gamma_1,&&|\gamma_1|\leq |\gamma_2|\\
\gamma_2,&&{\ |\gamma_1|> |\gamma_2|}
\end{matrix}\right.,
\qquad
\gamma_M:=\left\{ 
\begin{matrix}
\gamma_1,&&|\gamma_1|\geq |\gamma_2|\\
\gamma_2,&&{\ |\gamma_1|< |\gamma_2|}
\end{matrix}\right.,
\end{equation}}
and 
\be\label{Pdef}
P:=\left\{ 
\begin{matrix}
\alpha,&&|\alpha|\geq |\beta|\\
\beta,&&|\alpha|< |\beta|
\end{matrix}\right..
\ee


\begin{lemma}\label{L1}
For a qubit channel $\cN_{\phi,\theta}$ that is an extreme point in $\mathfrak{N}$, it is
\begin{align}\label{eqL1}
\max_{\psi\in{\cal P}(\mathbb{C}^2)}\|\Delta_{{{\cal N}_{\ {\phi, \theta}}}}(\psi)\|_1
&=\max_{0\leq s \leq 1} 2\sqrt{((1-s)\alpha +s\beta)^2+4s(1-s)\left[\left(\frac{\vert \gamma_1\vert
+\vert\gamma_2 \vert}{2}\right)^2-\alpha\beta\right]},\notag\\ \\
&=
  \left\{
  \begin{array}{lll}
    \frac {(|\gamma_1|+|\gamma_2|)\sqrt{(|\gamma_1|+|\gamma_2|)^2-4\alpha\beta}}{\sqrt{(|\gamma_1|
    +|\gamma_2|)^2-(\alpha+\beta)^2}}, & & \text{if}\;\; 
    \vert \alpha+\beta\vert < \vert \gamma_1\vert+\vert\gamma_2\vert \\ \\
  {2 |P|}  , & & \text{if}\;\; \vert \alpha+\beta\vert\geq \vert \gamma_1\vert+\vert\gamma_2\vert 
  \end{array}.
\right.\notag
\end{align}
\end{lemma}

\begin{proof}
 With the input state \eqref{qb}, taking into account \eqref{kmatrix}, we have
\begin{equation}
  \mathcal{N}_{\phi, \theta} (\psi)=K_0  \vert\psi\rangle\langle \psi \vert K_0^\dagger
  + K_1 \vert\psi\rangle\langle \psi \vert K_1^\dagger
  =\vert \psi'\rangle\langle \psi'\vert+\vert \psi''\rangle\langle \psi''\vert,
\end{equation}
where
\begin{equation}\label{re:14}
  \vert \psi'\rangle=a_0 \cos\theta\vert 0\rangle+a_1\cos\phi\vert 1\rangle, 
  \quad \vert \psi''\rangle=a_1 \sin\phi\vert 0\rangle+a_0\sin\theta\vert 1\rangle.
\end{equation}
Then, according to Eq.\eqref{dif1}, we will have:
\bea
\Delta_{{{\cal N}_{\ {\phi, \theta}}}}(\psi)
=\Delta_{{\psi'}}+\Delta_{{\psi''}},
\eea
with the quantities defined, in the canonical basis, as:
\be\label{delpsi}
\Delta_{{\psi'}}:=
\begin{pmatrix}
\vert a_0\vert^2\alpha & a_0a_1^*\gamma_1\\  a_0^*a_1\gamma_1&\vert a_1\vert^2 \beta
\end{pmatrix},
\quad\quad
\Delta_{{\psi''}}=
\begin{pmatrix}
-\vert a_1\vert^2 \beta & a_0^*a_1 \gamma_2 \\ a_0a_1^* \gamma_2 & -\vert a_0\vert^2 \alpha
\end{pmatrix}.
\ee
It follows that 
\bea
 \|\Delta_{{{\cal N}_{\ {\phi, \theta}}}}(\psi) \|_1
 &=& 2 \sqrt{(|a_0|^2\alpha-|a_1|^2\beta)^2
 +\vert a_0a_1^*\gamma_1+a_0^*a_1\gamma_2\vert^2}.
\eea
Without loss of generality, we assume $a_0\in\mathbb{R}$ and $a_1=|a_1| e^{i\eta}\in\mathbb{C}$. Taking into account the normalization condition, and the fact that $\vert \gamma_1+e^{2i\eta}\gamma_2 \vert^2$
reaches the maximum $\left( |\gamma_1|+|\gamma_2|\right)^2$,
we are left with finding 
\begin{equation}\label{re:33bis}
\max_{0\leq |a_1|\leq 1} 2\sqrt{(\alpha-(\alpha+\beta)|a_1|^2)^2+|a_1|^2(1-|a_1|^2)(\vert \gamma_1\vert
+\vert\gamma_2 \vert)^2},
\end{equation}
giving Eq.\eqref{eqL1} for
\begin{equation}\label{sol1}
  \left\{
  \begin{array}{lll}
    \vert a_1\vert^2=\frac{2\alpha(\alpha+\beta)-(|\gamma_1|+|\gamma_2|)^2}{2(\alpha+\beta)^2
    -2(|\gamma_1|+|\gamma_2|)^2}, 
    & & \hbox{if}\; \vert \alpha+\beta\vert < \vert \gamma_1\vert+\vert\gamma_2\vert \\
    \vert a_1\vert^2=0, & & \hbox{if}\; \vert \alpha+\beta\vert\geq \vert \gamma_1\vert+\vert\gamma_2\vert 
    \;\hbox{and}\; |\alpha| \geq |\beta|\\
     \vert a_1\vert^2=1, & & \hbox{if}\; \vert \alpha+\beta\vert\geq \vert \gamma_1\vert+\vert\gamma_2\vert 
     \;\hbox{and}\; |\alpha|\leq |\beta|
  \end{array}
\right..
\end{equation}
\end{proof}

\bigskip

\begin{lemma}\label{L3}
For a qubit channel ${\cal N}_{\phi, \theta}$ that is an extreme point in $\mathfrak{N}$,
in finding $\max_{{\Psi\in{\cal P} (\mathbb{C}^2\otimes\mathbb{C}^2)}}
  \Vert \Delta_{\operatorname{id} \otimes {{\cal N}_{\ {\phi, \theta}}}} (\Psi)\Vert_1$ 
 it is enough to search in the two dimensional subspace $span\{\vert 00\rangle,\vert 11\rangle\}$ or $span\{\vert 01\rangle,\vert 10\rangle\}$, and it results
\begin{align}
  &\max_{{\Psi\in{\cal P} (\mathbb{C}^2\otimes\mathbb{C}^2)}}
  \Vert \Delta_{\operatorname{id} \otimes {{\cal N}_{\ {\phi, \theta}}}} (\Psi)\Vert_1\notag\\
   &= \max_{0\leq s\leq 1} \frac{1}{2} \sum_{i, j=1}^{2}
   \Bigg\vert (1-s) \alpha + s \beta 
   +(-1)^i\sqrt{\left((1-s) \alpha + s \beta\right)^2+4s(1-s)(\gamma_j^2-\alpha\beta)}\Bigg\vert.
   \label{eq35}
\end{align}
\end{lemma}

\begin{proof}
First, we find an upper bound for
$\Vert \Delta_{\operatorname{id}\otimes{{\cal N}_{\ {\phi, \theta}}}} (\Psi )\Vert_1$, 
where $\ket \Psi$ is given in Eq.\eqref{2qb}:
\begin{align}
  \Vert \Delta_{\operatorname{id}\otimes {{\cal N}_{\ {\phi, \theta}}}} (\Psi)\Vert_1&=\left\Vert \operatorname{id}\otimes \left(K_0^{(1)} K_0^{(1)^\dagger}-K_0^{(2)} K_0^{(2)\dagger}\right) \Psi
  +\operatorname{id}
  \otimes \left(K_1^{(1)} K_1^{(1)\dagger}-K_1^{(2)} K_1^{(2)\dagger}\right)\Psi\right\Vert_1\nonumber\\
   &\leq \left\Vert \operatorname{id}\otimes \left(K_0^{(1)} K_0^{(1)\dagger}-K_0^{(2)}K_0^{(2)\dagger}\right) \Psi\right\Vert_1+\left\Vert\operatorname{id}\otimes \left( K_1^{(1)} K_1^{(1)\dagger}-K_1^{(2)} K_1^{(2)\dagger}\right)\Psi\right\Vert_1\nonumber\\
   &= \Vert C\Vert_1+ \Vert D\Vert_1\label{re:26},
\end{align}
where the superscripts refer to the two channels and the
matrices $C$ and $D$ are defined as
\begin{align}
C&:=\left(
  \begin{array}{cccc}
    \alpha_1 \vert a_{00}\vert^2 & \gamma_1a_{00}a_{01}^* & \alpha_1a_{00}a_{01}^* & \gamma_1a_{00}a_{11}^* \\
    \gamma_1 a_{00}^*a_{01} & \beta_1\vert a_{01}\vert^2 & \gamma_1a_{01}a_{10}^* & \beta_1a_{01}a_{11}^* \\
    \alpha_1a_{00}^*a_{10} & \gamma_1a_{01}^*a_{10} & \alpha_1\vert a_{10}\vert^2 & \gamma_1a_{10}a_{11}^* \\
    \gamma_1a_{00}^*a_{11} & \beta_1a_{01}^*a_{11} & \gamma_1a_{10}^*a_{11} & \beta_1\vert a_{11}\vert^2 \\
  \end{array}
\right),\\ \notag \\
D&:=\left(
  \begin{array}{cccc}
    \alpha_2 \vert a_{01}\vert^2 & \gamma_2a_{00}^*a_{01} & \alpha_2a_{01}a_{11}^* & \gamma_2a_{01}a_{10}^* \\
    \gamma_2 a_{00}a_{01}^* & \beta_2\vert a_{00}\vert^2 & \gamma_2a_{00}a_{11}^* & \alpha_2a_{00}a_{10}^* \\
    \alpha_2a_{01}^*a_{11} & \gamma_2a_{00}^*a_{11} & \alpha_2\vert a_{11}\vert^2 & \gamma_2a_{10}^*a_{11} \\
    \gamma_2a_{01}^*a_{10} & \alpha_2a_{00}^*a_{10} & \gamma_2a_{10}a_{11}^* & \beta_2\vert a_{10}\vert^2 \\
  \end{array}
\right).
\end{align}
The non-zero eigenvalues of $C$ and $D$ are respectively:
\begin{multline}\label{re:41}
\frac{1}{2}\left(\alpha(|a_{00}|^2+|a_{10}|^2)+\beta(|a_{11}|^2+|a_{01}|^2)\right)\\
  \pm\frac{1}{2}\sqrt{\left((|a_{00}|^2+|a_{10}|^2)\alpha-(|a_{01}|^2+|a_{11}|^2)\beta\right)^2
  +4\gamma_1^2(\vert a_{00}\vert^2+\vert a_{10}\vert^2)(\vert a_{01}\vert^2+\vert a_{11}\vert^2)},
\end{multline}
and
\begin{multline}\label{re:42}
 \frac{1}{2}\left(-\beta(|a_{11}|^2+|a_{01}|^2)-\alpha(|a_{00}|^2+|a_{10}|^2)\right)\\
  \pm\frac{1}{2}\sqrt{\left(-(|a_{00}|^2+|a_{10}|^2)\beta+(|a_{01}|^2+|a_{11}|^2)\alpha\right)^2
  +4\gamma_2^2(\vert a_{00}\vert^2+\vert a_{10}\vert^2)(\vert a_{01}\vert^2+\vert a_{11}\vert^2)}.
\end{multline}
Expressions \eqref{re:41} and \eqref{re:42} imply that, in order to optimize the upper bound \eqref{re:26}, it is enough to consider states in one of the following two dimensional subspaces: $span\{\vert 00\rangle,\vert 01\rangle\}$, $span\{\vert 00\rangle,\vert 11\rangle\}$, $span\{\vert 10\rangle,\vert 01\rangle\}$,
$span\{\vert 10\rangle,\vert 11\rangle\}$.

\bigskip

Next, we show that the subspace we have to consider for the optimization must be 
either $span\{\vert 00\rangle,\vert 11\rangle\}$, or $span\{\vert 01\rangle,\vert 10\rangle\}$. This is done by showing that in such cases the upper bound \eqref{re:26} is actually achievable.\\
Let us take $\ket {\Phi}=a_{0}\ket{00}+a_{1}\ket{11}$, with $|a_{0}|^2+|a_{1}|^2=1$,
then we can find that
\begin{equation}\label{eq:idNPhi}
 {\rm id}\otimes \mathcal{N}_{\phi,\theta}(\Phi)=K_0 \ket{\Phi}\bra {\Phi}K_0^\dagger+ K_1 \ket{\Phi}\bra {\Phi} K_1^\dagger=\ket{\Phi'}\bra{\Phi'}+\ket{\Phi''}\bra{\Phi''},
\end{equation}
where
\begin{equation}\label{re:12}
  \vert \Phi'\rangle:=a_{0} \cos\theta\ket{00}+a_{1}\cos\phi\ket{11}, \quad 
  \vert \Phi''\rangle:=a_{0} \sin\theta\ket{01}+a_{1}\sin\phi\ket{10}.
\end{equation}
It is clear that  $\vert \Phi'\rangle \perp \vert \Phi''\rangle$. Therefore, it results
\begin{equation}\label{re:orthogonal}
  \Vert {\rm id}\otimes \mathcal{N}(\Phi)\Vert_1
  =\Vert \vert \Phi'\rangle\langle \Phi'\vert\Vert_1+\Vert\vert \Phi''\rangle\langle \Phi''\vert\Vert_1.
\end{equation}
As a consequence we have that $\Phi$ allows us to saturate the bound 
\eqref{re:26}, namely
\be\label{delphi}
  \Vert \Delta_{\operatorname{id}\otimes{{\cal N}_{\ {\phi, \theta}}}} (\Phi)\Vert_1
  = \| \Delta_{{\Phi'}}\|_1+\| \Delta_{{\Phi''}}\|_1.
\ee
The same result can be found starting from $|\Phi\rangle=a_0 |00\rangle+a_1|11\rangle$.

In the canonical basis the quantities at the l.h.s. of Eq.\eqref{delphi} read:
\be\label{delphimatrix}
\Delta_{{\Phi'}}:=
\begin{pmatrix}
\vert a_{0}\vert^2\alpha  &0  & 0   & a_0a_1^*\gamma_1\\
 0  &  0  & 0   & 0  \\
  0  &  0  & 0   & 0  \\
a_{0}^*a_{1}\gamma_1&  0 &    0  & \vert a_1\vert^2 \beta
\end{pmatrix},
\quad\quad
\Delta_{{\Phi''}}:=
\begin{pmatrix}
 0  &  0  & 0   & 0  \\
0 &  -\vert a_1\vert^2 \alpha & a_0a_1^* \gamma_2 &0 \\
0 & a_0^*a_1 \gamma_2 & -\vert a_0\vert^2 \beta &0 \\
 0  &  0  & 0   & 0
\end{pmatrix}.
\ee
Comparing Eq.\eqref{delphimatrix} with Eq.\eqref{delpsi} we deduce that
\bea\label{rel}
 \Vert \Delta_{\operatorname{id}\otimes {{\cal N}_{\ {\phi, \theta}}}} (\Phi)\Vert_1&=&
   \| \Delta_{{\Phi'}}\|_1+\| \Delta_{{\Phi''}}\|_1\\
   &=& \| \Delta_{\psi'}\|_1+\| \Delta_{{\psi''}}\|_1\notag\geq\|\Delta_{{{{\cal N}_{\ {\phi, \theta}}}}}(\psi)\|_1.
 \eea
Finally, from Eqs.\eqref{delphi} and \eqref{delphimatrix}, we derive
\begin{align}
  &\max_{{\Psi\in{\cal P} (\mathbb{C}^2\otimes\mathbb{C}^2)}}
  \Vert \Delta_{\operatorname{id} \otimes {{\cal N}_{\ {\phi, \theta}}}} (\Psi)\Vert_1
  =\max_{\Phi\in {\cal P} (\mathbb{C}^2) }\Vert  \Delta_{\operatorname{id}\otimes{{\cal N}_{\ {\phi, \theta}}}} (\Phi) \Vert_1\notag\\
   &= \max_{0\leq s\leq 1} \frac{1}{2}  \sum_{i, j=1}^{2}
   \left\vert (1-s) \alpha+ s \beta +(-1)^i\sqrt{\left((1-s) \alpha + s \beta\right)^2+4s(1-s)(\gamma_j^2-\alpha\beta)}\right\vert.
\end{align}

\end{proof}


\begin{oss}\label{observation}

First notice that, being $(\gamma_j^2-\alpha\beta)\geq 0$ for $j=1,2$, Eq.\eqref{eq35} reduces to 
\begin{align}\label{eqL2}
\max_{\Psi\in{\cal P}(\mathbb{C}^2\otimes\mathbb{C}^2)}\Vert \Delta_{\operatorname{id} 
\otimes {{\cal N}_{{\phi, \theta}}}} (\Psi)\Vert_1
=\max_{0\leq s \leq 1}&\left\{\sqrt{\left((1-s) \alpha + s \beta\right)^2+4s (1-s) (\gamma_1^2-\alpha\beta)}\right.\notag\\
&+\left.\sqrt{\left((1-s) \alpha + s \beta\right)^2+4s (1-s) (\gamma_2^2-\alpha\beta)}\right\}.
\end{align}
Then, observe that by means of Eqs.\eqref{eqL2} and \eqref{eqL1},  we can write
\begin{align}
  \max_{\Psi\in{\cal P} (\mathbb{C}^2\otimes\mathbb{C}^2)}
  \Vert \Delta_{\operatorname{id} \otimes {{\cal N}_{{\phi, \theta}}}} (\Psi)\Vert_1
  -\max_{\psi\in{\cal P} (\mathbb{C}^2)}\Vert \Delta_{{{\cal N}_{{\phi, \theta}}}} (\psi )\Vert_1
  &= \max_{0\leq s\leq 1}f\left(s\right)-\max_{0\leq s\leq 1}g\left(s\right),
  \label{eq:diff}
\end{align}
where
\begin{align}
f\left(s\right)&:=
\sum_{i=1}^{2}\sqrt{\left((1-s) \alpha - s \beta\right)^2+4s (1-s) \gamma_i^2} \label{fdef},\\
g\left(s\right)&:= 2\sqrt{((1-s)\alpha-s\beta)^2+4s(1-s) \left(\frac{\vert \gamma_1\vert +\vert\gamma_2\vert}{2}\right)^2}.
\label{gdef}
\end{align}
It follows that:
\begin{itemize}
\item When $|\gamma_1|=|\gamma_2|$, it is $f\left(s\right)=g(s)$, implying that side entanglement is not useful.

\item When $|\gamma_1|\neq|\gamma_2|$, suppose that $\max_{0\leq s\leq 1} g\left(s\right)= g(s_*)$, then, using relation \eqref{rel}, we will have
\be
\max_{0\leq s\leq 1}f\left(s\right)\geq f\left(s_*\right)\geq g(s_*).
\ee
In other words
\begin{itemize}
\item If  $ f\left( s_*\right) \neq g(s_*) \Rightarrow \max_{\Psi}\Vert \Delta_{\operatorname{id} \otimes {{\cal N}_{{\phi, \theta}}}} (\Psi)\Vert_1>\max_{\psi}\Vert \Delta_{{{\cal N}_{{\phi, \theta}}}} (\psi )\Vert_1 $
$\Rightarrow $ side entanglement is useful.

 \item If  $ f\left( s_*\right) = g(s_*)$ and $s_* \neq {\rm argmax} \, f\left(s\right) \Rightarrow $ side entanglement is useful.

\item
 If  $ f\left( s_*\right) = g(s_*)$ and $s_*={\rm argmax} \, f\left(s\right) \Rightarrow $ side entanglement is not useful.

\end{itemize}
\end{itemize}
\end{oss}

With this in mind we can now state the following Theorem.


\begin{theorem}\label{thm:sidentangelmentnessery}
In the discrimination between two qubit channels, 
${\cal N}_{\phi_1, \theta_1}$ and ${\cal N}_{\phi_1, \theta_1}$,
that are extreme point maps of $\mathfrak{N}$,  when $|\gamma_1|=|\gamma_2|$ side entanglement is not useful. 
Instead, when $|\gamma_1|\neq|\gamma_2|$ the following exhausts all possibilities:
\begin{itemize}
\item[A.] Side entanglement is not useful when
\begin{enumerate}
\item[1)]  ${ 2|\gamma_M| }\leq|\alpha+\beta|$;
\item[2)] $|\gamma_1|+|\gamma_2|\leq|\alpha+\beta|<{2|\gamma_M|}$ and
{$P\left(\alpha+\beta\right)\geq\gamma_1^2+\gamma_2^2$;}
\item[3)] $|\alpha+\beta|<|\gamma_1|+|\gamma_2|$ and $\alpha\neq\beta$;
\item[4)] $|\alpha+\beta|<|\gamma_1|+|\gamma_2|$, $\alpha=\beta$ and $\alpha^2 \leq |\gamma_1||\gamma_2|$.
\end{enumerate}
\item[B.] Side entanglement is useful when
\begin{enumerate}
\item[1)] $|\gamma_1|+|\gamma_2|\leq|\alpha+\beta|<2|\gamma_M|$ and 
{$P\left(\alpha+\beta\right)<\gamma_1^2+\gamma_2^2$;}
\item[2)] $|\alpha+\beta|<|\gamma_1|+|\gamma_2|, \alpha=\beta$, 
and $\alpha^2>|\gamma_1||\gamma_2|$.
\end{enumerate}
\end{itemize}
\end{theorem}

\begin{proof} 
When $|\gamma_1|=|\gamma_2|$, the functions we are optimizing become the same, i.e., $f\left(s\right)= g(s)$, hence side entanglement  cannot be useful as put forward in Observation \ref{observation}. 

\bigskip

$\bullet$ If ${2|\gamma_M|}\leq \vert \alpha+\beta\vert $, 
we have from Lemma \ref{L1},
\begin{align}
  {\rm Eq.}\eqref{eq:diff}&=\max_{0\leq s\leq 1}\sum_{i=1}^{2}\sqrt{\left((1-s) \alpha -s\beta\right)^2
  +4s (1-s) \gamma_i^2}-  {2 |P|} \notag\\
  &\leq \sum_{i=1}^{2}\max_{0\leq s\leq 1}\sqrt{\left((1-s) \alpha - s \beta\right)^2
  +4s (1-s) \gamma_i^2}-  {2 |P|} \notag\\
  &\leq \max\{|\alpha|,|\beta|\}+\max\{|\alpha|,|\beta|\} - {2 |P|}=0,
\end{align}
meaning that side entanglement is not useful. This proves A.1).

\bigskip

$\bullet$  When $\vert\gamma_1\vert+\vert \gamma_2\vert\leq\vert \alpha+\beta\vert 
<{2|\gamma_M|}$, if there exist a value $s$, $0<s<1$, for which
$f(s)>g(s_*)$, then side entanglement is useful. 
Thanks to Lemma \ref{L1} we can rewrite $f(s)>g(s_*)$ as
\be\label{c1}
\sum_{i=1}^{2}\sqrt{\left(s \alpha -(1- s) \beta\right)^2+4s(1-s) \gamma_i^2}> {2 |P|}.
\ee
Instead, if $\forall s\in \left(0,1\right)$ Eq.\eqref{c1} does not hold, then side entanglement is not useful.

The two above alternatives are equivalent to the fact that the following function
\bea
F(s)&:=&P^2(P^2-\beta^2)+2P^2\left( \beta(\alpha+\beta)-(\gamma_1^2+\gamma_2^2)\right)s\notag\\
&+&\left[P^2\left( 2(\gamma_1^2+\gamma_2^2)-(\alpha+\beta)^2\right)+(\gamma_1^2-\gamma_2^2)^2\right]s^2\notag\\
&-&2(\gamma_1^2-\gamma_2^2)^2 s^3+(\gamma_1^2-\gamma_2^2)^2 s^4,
\eea
can be negative for some values of $s$, or it is always not negative.

We then distinguish two cases:
\begin{itemize}
\item[i)] $ P=\alpha$.

\begin{itemize}
\item If $\alpha(\alpha+\beta)<\gamma_1^2+\gamma_2^2$, we can make $F(s)<0$ with $s$ arbitrarily close to $0$. Hence, side entanglement is useful. This partly proves B.1).
\item If $ \alpha (\alpha+\beta)\geq\gamma_1^2+\gamma_2^2$, for each $0<s<1$, we have
  \begin{equation}
    2\alpha(\alpha+\beta)-(\alpha+\beta)^2(1-s)-2s(\gamma_1^2+\gamma_2^2)\geq 0,
  \end{equation}
which implies that $F(s)\geq0$ and so, side entanglement is not useful. This partly proves A.2).
\end{itemize}

\item[ii)] $P=\beta$.

\begin{itemize}
\item If $\beta(\alpha+\beta)<\gamma_1^2+\gamma_2^2$, we can make $F(s)<0$, with $s$ arbitrarily close to $1$. Hence side entanglement is useful. This completes the proof of B.1).
\item If $ \beta (\alpha+\beta)\geq\gamma_1^2+\gamma_2^2$ 
for each $0<s<1$, we have
  \begin{equation}
    2\beta(\alpha+\beta)-(\alpha+\beta)^2(1-s)-2s(\gamma_1^2+\gamma_2^2)\geq 0,
  \end{equation}
which implies that $F(s)\geq0$ and so, side entanglement is not useful. This completes the proof of A.2).
\end{itemize}
\end{itemize}

\bigskip

$\bullet$ When $|\alpha+\beta|<|\gamma_1|+|\gamma_2|$, it is
\bea\label{neqequiv}
 f\left(s_*\right) \neq g(s_*) 
\Leftrightarrow 
\left((1-s_*) \alpha - s_* \beta\right) \left( |\gamma_1|-|\gamma_2|\right)^2\neq 0,
\eea
for any $0<s_*<1$.
Inserting the solution \eqref{sol1} into \eqref{neqequiv}, we can realize that 
$ \left((1-s_*) \alpha - s_* \beta\right) \left( |\gamma_1|-|\gamma_2|\right)^2 = 0$ for
$\alpha\neq\beta$.
Hence, in this case side entanglement is not useful. This proves A.3).

\bigskip

$\bullet$ When $\vert \alpha+\beta\vert < \vert \gamma_1\vert+\vert\gamma_2\vert$ and $\alpha=\beta$, the solution \eqref{sol1} becomes $s_*=\frac{1}{2}$. One can easily check that this is also an extrema of the function $f$, but it turns out to be the maximum only when $\alpha^2\leq|\gamma_1||\gamma_2|$. 
It means that if 
$\alpha^2>|\gamma_1||\gamma_2|$, side entanglement is useful, proving B.2). Instead 
if $\alpha^2\leq|\gamma_1||\gamma_2|$ side entanglement is not useful, proving A.4).

\end{proof}

\begin{remark}
The channels of Theorem \ref{thm:sidentangelmentnessery} are on the faces of the tetrahedron and correspond to ellipsoids that have one point in common with the boundary of the Bloch sphere. Examples of this kind of channels are those having two Kraus operators that are non-unitary, like the amplitude damping channel that can be written as $ \mathcal{N}_{\phi, \theta=0}(\rho)$. For it, from Theorem \ref{thm:sidentangelmentnessery}, we can easily recover the results of Ref.\cite{MilaJPA}.
\end{remark}


\section{Discriminating between non-extremal qubit channels}\label{sec:non-extreme}

In this Section we extend the result of Theorem  \ref{thm:sidentangelmentnessery}  to the discrimination between two non-extremal qubit channels. These can be written as convex combination of extremal channels as follows:
\begin{align}
{\cal N}_{\{\lambda_1,\phi_1,\theta_1,\phi_1',\theta_1'\}} &=\lambda_1\cN_{\phi_1,\theta_1}+(1-\lambda_1)\cN_{\phi_1',\theta_1'}, \label{N1}\\
{\cal N}_{\{\lambda_2,\phi_2,\theta_2,\phi_2',\theta_2'\}}&=\lambda_2\cN_{\phi_2,\theta_2}+(1-\lambda_2)\cN_{\phi_2',\theta_2'}, \label{N2}
\end{align}
where $0< \lambda_1, \lambda_2< 1$.

Then, we redefine the parameters of \eqref{re:18}, \eqref{re:19}, \eqref{re:20} and \eqref{gamma2} 
in a more general fashion as
\begin{align}
  \alpha &:=(\lambda_1\cos^2\theta_1+(1-\lambda_1)\cos^2\theta'_1 )-(\lambda_2\cos^2\theta_2+(1-\lambda_2)\cos^2\theta'_2)\label{al2}, \\
     \beta&:=(\lambda_1\cos^2\phi_1+(1-\lambda_1)\cos^2\phi'_1 )-(\lambda_2\cos^2\phi_2+(1-\lambda_2)\cos^2\phi'_2) \label{be2},\\
   \gamma_1&:=\lambda_1\cos\phi_1\cos\theta_1+(1-\lambda_1)\cos\phi'_1 \cos\theta'_1-
   \lambda_2\cos\phi_2\cos\theta_2-(1-\lambda_2)\cos\phi'_2\cos\theta'_2 \label{g12},\\
    \gamma_2&:=\lambda_1\sin\phi_1\sin\theta_1+(1-\lambda_1)\sin\phi'_1\sin\theta'_1-
   \lambda_2\sin\phi_2\sin\theta_2-(1-\lambda_2)\sin\phi'_2\sin\theta'_2.\label{g22}
\end{align}

\begin{remark}\label{rem4}
{
Note that all the machinery of Lemmas \ref{L1}, \ref{L3} and Observation \ref{observation} applies 
to parameters \eqref{al2}-\eqref{g22} as well. At the core of this is the fact that, proceeding similarly to Eq.\eqref{eq:idNPhi}, we can now find
\bea
\operatorname{id} \otimes \mathcal{N}_{\{\lambda,\phi,\theta,\phi',\theta'\}}(\Phi)
=\lambda_1\left(\ket {\Phi'}\bra{\Phi'}+\ket {\Phi''}\bra{\Phi''}\right)+(1-\lambda_1)\left(\ket {\Gamma'}\bra{\Gamma'}+\ket {\Gamma''}\bra{\Gamma''}\right),
\eea
where $\ket {\Phi'}$, $\ket {\Phi''}$ are as in Eq.\eqref{re:12}, while
$\ket {\Gamma'}:=a_0\cos\theta_1'\ket{00}+a_1\cos\phi_1'\ket{11}$, 
$\ket {\Gamma''}:=a_0\sin\theta_1'\ket{01}+a_1\sin\phi_1'\ket{10}$.
Then, since $\{\ket {\Phi'},\ket {\Gamma'} \}\perp\{\ket {\Phi''},\ket {\Gamma''} \}$, 
we can draw a conclusion analogous to Eq.\eqref{re:26}, namely
\be\label{eq60}
\Vert \Delta_{\operatorname{id} \otimes \mathcal{N}} (\Phi)\Vert_1=\Vert\lambda_1\Delta_{\Phi'}+(1-\lambda_1)\Delta_{\Gamma'}\Vert_1+\Vert\lambda_1\Delta_{\Phi''}+(1-\lambda_1)
\Delta_{\Gamma''}\Vert_1.
\ee

However, parameters \eqref{al2}-\eqref{g22} } no longer guarantee the satisfaction of conditions
$(\gamma_j^2-\alpha\beta)\geq 0$, for $j=1,2$, as in Sec.\ref{sec:extreme}. 
Whenever $\gamma_m^2 \geq \alpha\beta$, this holds true and 
 Theorem \ref{thm:sidentangelmentnessery} is directly applicable.   
For the other cases we have the following results.
\end{remark}

\begin{theorem}\label{thm:nonextr}
In the discrimination between two qubit channels,
${\cal N}_{\{\lambda_1,\phi_1,\theta_1,\phi_1',\theta_1'\}}$
and ${\cal N}_{\{\lambda_2,\phi_2,\theta_2,\phi_2',\theta_2'\}}$,
that are not extreme point maps of $\mathfrak{N}$, when $\gamma_M^2\leq \alpha \beta$
side entanglement is not useful. Instead, when 
$\gamma_m^2<\alpha\beta<\gamma_M^2$ 
the following exhausts all possibilities:
\begin{itemize}
\item[A.] Side entanglement is not useful when
\begin{enumerate}

\item[1)] 
$|\gamma_1|+|\gamma_2|>|\alpha +\beta|$, while $\alpha=\beta$ and $|\gamma_m|=|\alpha|$;

\item[2)]
$|\gamma_1|+|\gamma_2|>|\alpha +\beta|$ while $\alpha\neq\beta$, 
$ 2\gamma_M(\alpha+\beta)=(|\gamma_1|+|\gamma_2|)^2$, 
 and
$\gamma_M=\frac{1+\alpha\beta\pm\sqrt{(\alpha^2-1)(\beta^2-1)}}{\alpha+\beta}$;

\item[3)] 
$|\gamma_1|+|\gamma_2|\leq|\alpha +\beta|$  and $|\gamma_M|\leq  {|P|}$.
\end{enumerate}

\item[B.] Side entanglement is useful when

\begin{enumerate}
\item[1)] 
$|\gamma_1|+|\gamma_2|>|\alpha +\beta|$, while $\alpha=\beta$ and $|\gamma_m|\neq|\alpha|$;

\item[2)] 
$|\gamma_1|+|\gamma_2|>|\alpha +\beta|$, while $\alpha\neq\beta$ and 
$ 2\gamma_M(\alpha+\beta)\neq(|\gamma_1|+|\gamma_2|)^2$;

\item[3)] $|\gamma_1|+|\gamma_2|>|\alpha +\beta|$, while $\alpha\neq\beta$, 
$ 2\gamma_M(\alpha+\beta)=(|\gamma_1|+|\gamma_2|)^2$,
 and 
$\gamma_M\neq\frac{1+\alpha\beta\pm\sqrt{(\alpha^2-1)(\beta^2-1)}}{\alpha+\beta}$;

\item[4)] 
$|\gamma_1|+|\gamma_2|\leq|\alpha +\beta|$ and $|\gamma_M|>{|P|}$.
\end{enumerate}

\end{itemize}
\end{theorem}

\begin{proof} 
If $\gamma_m^2\leq \alpha \beta$, on the one hand we have from Eq.\eqref{eq35}
$\max_{\Psi}\Vert \Delta_{\operatorname{id} \otimes \mathcal{N}} (\Psi)\Vert_1
=\max_{0\leq s \leq 1}\vert (1-s) \alpha+s\beta   \vert$,
on the other hand from Lemma \ref{L1} we have
$\max_{\psi}\Vert \Delta_{\mathcal{N}} (\psi )\Vert_1\geq\max_{0\leq s\leq 1}\vert (1-s)\alpha+s\beta   \vert$.
Thus 
  \begin{align}
   \max_{\Psi}\Vert \Delta_{\operatorname{id} \otimes \mathcal{N}} (\Psi)\Vert_1
   -\max_{\psi}\Vert \Delta_{\mathcal{N}} (\psi )\Vert_1
   &\leq
   \max_{0\leq s\leq 1}\vert (1-s)\alpha+s\beta   \vert
   -\max_{0\leq s\leq 1}\vert (1-s)\alpha+s\beta   \vert=0,
  \end{align}
meaning that side entanglement is not useful.

\bigskip  

The rest of the proof is devoted to the case $\gamma_m^2<\alpha\beta<\gamma_M^2$. 
In such a case Eq.\eqref{eq35} reduces to
  \begin{align}
   \max_{\Psi}\Vert \Delta_{\operatorname{id} \otimes \mathcal{N}} (\Psi)\Vert_1
   =\max_{0\leq s\leq 1}\left( |(1-s)\alpha+s\beta|+\sqrt{\left( (1-s) \alpha + s \beta\right)^2
   +4 s(1-s) (\gamma_M^2-\alpha\beta)}\right)\label{re:63}.
  \end{align}
  Let us define
  \be\label{Gdef}
 G(s):= |(1-s)\alpha+s\beta|+\sqrt{\left((1-s)\alpha +s \beta\right)^2+4s(1-s) (\gamma_M^2-\alpha\beta)}.
  \ee
This function has maxima for the argument $s_*=\min\{\max\{0,\tilde{s}\},1\}$, where
  \be
  \tilde{s}=\left\{
  \begin{matrix}  
  \frac{\gamma_M-\alpha}{2\gamma_M-(\alpha+\beta)}, && \text{when} \quad \gamma_M (\alpha+\beta)>0\\ \\
  \frac{\gamma_M+\alpha}{2\gamma_M+(\alpha+\beta)}, && \text{when} \quad \gamma_M (\alpha+\beta)<0
  \end{matrix}
  \right.,\label{sstar}
  \ee
 leading to 
 \be
 \max_{\Psi}\Vert \Delta_{\operatorname{id} \otimes \mathcal{N}} (\Psi)\Vert_1=\left\{
  \begin{matrix}
 |\gamma_M|+\left| \frac{2 \alpha  \beta -\gamma_M (\alpha +\beta )}{\alpha +\beta -2 \gamma_M}\right|&& \text{when}\quad  \gamma_M(\alpha+\beta)>0 \\ \\
 |\gamma_M|+\left| \frac{2 \alpha  \beta+ \gamma_M (\alpha +\beta )}{\alpha +\beta +2 \gamma_M}\right|&& 
 \text{when}\quad \gamma_M(\alpha+\beta)<0
  \end{matrix}
  \right..
 \ee
 
\bigskip
 
$\bullet$ Consider now $|\gamma_1|+|\gamma_2|>|\alpha +\beta|$. We distinguish the following cases:
 \begin{enumerate}
 \item When  $\alpha\neq\beta$ and $ 2\gamma_M(\alpha+\beta)\neq(|\gamma_1|+|\gamma_2|)^2$, using \eqref{sol1} and \eqref{sstar}, we have $s_*\neq |a_1|^2$, meaning that side entanglement is useful. This proves B.2).
 
 
 \item When $\alpha\neq\beta$ and $2\gamma_M(\alpha+\beta)=(|\gamma_1|+|\gamma_2|)^2$ we have 
$\gamma_M (\alpha+\beta)>0$, and $ s^*=|a_1|^2=\frac{\gamma_M+\alpha}{2\gamma_M+(\alpha+\beta)}$ 
which yields
\begin{align}
\frac{\max_{\Psi}\Vert \Delta_{\operatorname{id} \otimes \mathcal{N}} (\Psi)\Vert_1}{\max_{\psi}\Vert \Delta_{\mathcal{N}} (\psi )\Vert_1}
&=\frac{1}{2}\left( \sqrt{\frac{\alpha +\beta -2\gamma_M}{\gamma_M\left(2 \alpha  \beta -\gamma_M(\alpha +\beta)\right) }}+\sqrt{\frac{\gamma_M\left(2 \alpha  \beta -\gamma_M(\alpha +\beta)\right) }{\alpha +\beta -2\gamma_M}}\right).
\label{FRAC}
\end{align}
Side entanglement is useful when \eqref{FRAC}$>1$. This happens for
 $\gamma_M\neq\frac{1+\alpha\beta\pm\sqrt{(\alpha^2-1)(\beta^2-1)}}{\alpha+\beta}$, proving B.4).
In contrast, for $\gamma_M=\frac{1+\alpha\beta\pm\sqrt{(\alpha^2-1)(\beta^2-1)}}{\alpha+\beta}$ 
it results \eqref{FRAC}${=}1$ and side entanglement is useless. This proves A.2).

 
\item When $\alpha=\beta$, again using \eqref{sol1} and \eqref{sstar}, we have $s_*=|a_{1}|^2=\frac{1}{2}$
leading to 
\bea
G\left(\frac{1}{2}\right)=|\gamma_M|+|\alpha| \quad \text{and} \quad
g\left(\frac{1}{2}\right)=| \gamma_M| +| \gamma_m|, 
\eea
with $G$, $g$ defined respectively in \eqref{Gdef}, \eqref{gdef}.
Thus, we have that
\begin{itemize}
\item[-] for $|\gamma_m|\neq|\alpha|$ it is $G(1/2)\neq g(1/2)$ and side entanglement is useful.
This proves B.1);

\item[-] for $|\gamma_m|=|\alpha|$ it is $G(1/2)=g(1/2)$ and side entanglement is not useful.
This proves A.1).
\end{itemize}

\end{enumerate}

\bigskip

$\bullet$ Finally, when $|\gamma_1|+|\gamma_2|\leq|\alpha +\beta|$, 
if there exist a value $s$, $0<s<1$ for which
$G(s)>g(s_*)$, then side entanglement is useful. 
Actually we can rewrite $G(s)>g(s_*)$ as
\be\label{eq69}
|(1-s)\alpha+s\beta|+\sqrt{\left((1-s) \alpha - s \beta\right)^2+4(s-s^2) \gamma_M^2}>2 {|P|}.
\ee
Instead, if $\forall s\in \left(0,1\right)$ Eq.\eqref{eq69} does not hold, then side entanglement is not useful.

The two above alternatives are equivalent to the fact that the following function 
\begin{align}
 R(s)&:=P^2( P^2-\alpha ^2)+2 P^2 \left(\alpha ^2-\gamma_M^2 \right)s \notag\\
 &+ \left[(\gamma_M^2 -\alpha  \beta )^2-P^2 \left(\alpha ^2+\beta ^2-2 \gamma_M^2 \right)\right]s^2\notag\\
&-2 s^3 (\gamma_M^2 -\alpha  \beta )^2+ (\gamma_M^2 -\alpha  \beta )^2 s^4,
\end{align}
can be negative for some values of $s$, or it is always not negative.

 \begin{itemize}
    \item[-] For $|\gamma_M|>{|P|}$, we can choose $s$ arbitrarily close to zero (or to 1) to make $R(s)<0$.
 This shows that side entanglement is useful and proves B.5).
    \item[-] For $|\gamma_M|\leq {|P|}$, for any $0{\leq} s\leq 1$, we have $R(s)>0$.
 It means that side entanglement is not useful, proving A.3).
  \end{itemize}
  
\end{proof}


\bigskip


\begin{remark}
The channels of Theorem \ref{thm:nonextr} are inside the tetrahedron and correspond to ellipsoids that have no point in common with the boundary of the Bloch sphere. Examples of this kind of channels are the quantum Pauli channels, that can be written as 
$ \lambda \mathcal{N}_{\theta, \theta}(\rho)+(1-\lambda)\mathcal{N}_{\theta', \pi+\theta'}(\rho)$.
For them, from Theorem \ref{thm:nonextr}, we can easily recover the results of Ref.\cite{Sacchi2005}.
\end{remark}


\section{Conclusion}\label{conclu}
{ In summary, we investigated the conditions under which side entanglement can effectively be used for channel discrimination and identified the specific circumstances in which it enhances, as well as those in which it does not enhance, the performance of discerning between two generic qubit channels, up to unitary pre- and post-processing. We present our findings in a decision tree, which is included in Appendix \ref{app}.

Our results were obtained through a constructive approach, whereby we first established the uselessness of side entanglement in discriminating quasi-extremal channels within the set $\mathfrak{N}$. We then analyzed the problem for extremal maps and finally extended our findings to generic channels that may be located within such a set.}

{ Our results can be straightforwardly extended to channels that include pre and post-processing unitaries, provided that these are diagonal or anti-diagonal { in the canonical basis.} This is because, such unitaries will preserve the two subspaces where the optimal input lives ($span\{\vert 00\rangle,\vert 11\rangle\}$, or $span\{\vert 01\rangle,\vert 10\rangle\}$) or map one into another. Furthermore, relations (24-25) and (38-39) will still hold true for newly defined parameters $\alpha,\beta, \gamma_1$ and $\gamma_2$.}

Several partial results previously obtained for specific kind of qubit channels can be recovered as particular cases of our general results. Examples are the discrimination of dephasing channel, amplitude damping channel and depolarizing channel.

From Theorems \ref{thm:sidentangelmentnessery} and \ref{thm:nonextr} appears that, quite generally, side entanglement is useful only in a limited portion of the parameters' space. Additionally, the profit is often brought by non-maximally entangled states. 

Our findings are based on the characterization of qubit channels in terms of affine maps on $\mathbb{R}^3$.
Geometric characterization of the action of quantum channels in higher dimensional systems would be useful for extending our results beyond qubit channels.

{\ 
Moving on from the framework of point-to-point quantum channel \cite{Gyo1}, one could
shift to more complex structures to discriminate, namely quantum networks \cite{Gyo2}. In contrast to channels, networks allow for intermediate access
points where information can be received, processed and reintroduced into the network.
A first attempt in this direction has been done in Ref. \cite{Hir21}
using the idea of quantum super-channels.
}


\appendix

\section{}\label{app}

\begin{figure}[H]
\centering
\begin{tikzpicture}[font=\footnotesize,thick]
 
\node[draw,
    rectangle,
    minimum height=1cm] (blocka1) {$\gamma_M^2\leq\alpha\beta$};
 \node[draw,
    rectangle,
    minimum width=3.5cm,
    minimum height=1cm,left of=blocka1,xshift=6.5cm] (blockb1) {$\gamma_m^2<\alpha\beta<\gamma_M^2$};
    \node[draw,
    rectangle,
    minimum height=1cm,left of=blockb1,xshift=7cm] (blockc1) {$\alpha\beta\leq\gamma_m^2$};
    { 
 
\node[draw,
    ellipse,text width=1cm,text centered,
    below=of blocka1] (blocka2) {Not useful};

\node[draw,
   rectangle,
    below=of blockb1,xshift=-0.2,yshift=0.2cm] (blockb2) {$|\gamma_1|+|\gamma_2|>|\alpha+\beta|$};

\node[draw,
   rectangle,
    below left=of blockb2,
   yshift=0.5cm,xshift=0.8cm] (blockb31){$\alpha=\beta$};
\node[draw,
   rectangle,
    below right=of blockb2,
xshift=-1.3cm,yshift=0.25cm] (blockb32) {$|\gamma_M|\leq |P|$};
   \node[draw,
   rectangle,
    below right=of blockb31,xshift=-0.6cm,yshift=-1.5cm] (blockb42){$|\gamma_m|=|\alpha|$};
\node[draw,
   rectangle,
    below left=of blockb31,text width=2.5cm,text centered,yshift=-0.5cm,
 xshift=1cm] (blockb41) {$2\gamma_M(\alpha+\beta)=(|\gamma_1|+|\gamma_2|)^2$};
  \node[draw,
    ellipse,text width=1cm,text centered,
    below left=of blockb32,xshift=1.25cm] (blockb43){Not useful};
\node[draw,
   ellipse ,text width=1cm,text centered,xshift=-1.5cm, yshift=0cm,
    below right=of blockb32] (blockb44) {Useful};
   
   \node[draw,
     ellipse,text centered,
    below left=of blockb41,xshift=1.4cm,yshift=-0.4cm] (blockb51) {Useful};
   
   \node[draw,
     rectangle,text centered,
    below right=of blockb41,
   xshift=-2.8cm, yshift=-2.4cm] (blockb52) {$\gamma_M=\frac{1+\alpha\beta\pm\sqrt{(\alpha^2-1)(\beta^2-1)}}{\alpha+\beta}$};
   
\node[draw,
    ellipse,text width=1cm,text centered,
    below left=of blockb42,xshift=1.4cm,yshift=-0.4cm] (blockb53) {Not useful};
   
   \node[draw,
     ellipse,text width=1cm,text centered,
    below right=of blockb42,
   xshift=-1.5cm, yshift=0cm] (blockb54) {Useful};
   \node[draw,
    ellipse,text width=1cm,text centered,
    below left=of blockb52,xshift=1cm,yshift=-0.4cm] (blockb61) {Not useful};
   
   \node[draw,
     ellipse,text width=1cm,text centered,yshift=-0.5cm,
    below =of blockb52,
   xshift=0cm, yshift=0cm] (blockb62) {Useful};
\node[draw,
   rectangle,xshift=0cm,
    below =of blockc1] (blockc2) {$|\gamma_1|=|\gamma_2|$};
\node[draw,
   rectangle,
    below left=of blockc2,xshift=2.2cm
  ,yshift=-2.5cm] (blockc31) {$|\gamma_1|+|\gamma_2|>|\alpha+\beta|$};
\node[draw,
  ellipse,text width=1cm,text centered, xshift=-1.25cm,
    below right=of blockc2] (blockc32) {Not useful};
    
    \node[draw,
   rectangle,
    below left=of blockc31,xshift=2cm
  ,yshift=-0.4cm] (blockc41) {$2|\gamma_M|>|\alpha+\beta|$};
\node[draw,
  rectangle,text width=1cm,text centered, xshift=-1cm,yshift=-0.6cm,
    below right=of blockc31] (blockc42) {$\alpha=\beta$};
    
    \node[draw,
   rectangle,
    below left=of blockc41,xshift=2.2cm
  ,yshift=-0.5cm] (blockc51) {$\gamma_1^2+\gamma_2^2>P(\alpha+\beta)$};
\node[draw,
  ellipse,text width=1cm,text centered, xshift=-1.25cm,
    below right=of blockc41] (blockc52) {Not useful};
    
      \node[draw,
   ellipse,
    below left=of blockc51,xshift=1.5cm
  ,yshift=-0.5cm] (blockc61) {Useful};
\node[draw,
  ellipse,text width=1cm,text centered, xshift=-1.25cm,
    below right=of blockc51] (blockc62) {Not useful};

    \node[draw,
   rectangle,
    below left=of blockc42,xshift=1.8cm
  ,yshift=-1.5cm] (blockc53) {$|\gamma_1||\gamma_2|<\alpha^2$};
\node[draw,
  ellipse,text width=1cm,text centered, xshift=-1.25cm,
    below right=of blockc42] (blockc54) {Useful};

      \node[draw,
   ellipse,
    below left=of blockc53,xshift=1.2cm
  ,yshift=-0.5cm] (blockc63) {Useful};
\node[draw,
  ellipse,text width=1cm,text centered, xshift=-1.25cm,
    below right=of blockc53] (blockc64) {Not useful};

\draw[-latex] (blocka1) --node[inner sep=1,fill=white] {No} (blockb1);

\draw[-latex] (blockb1)  --node[inner sep=1,fill=white] {No} (blockc1);

\draw[-latex] (blocka1)  --node[inner sep=1,fill=white] {Yes} (blocka2);

\draw[-latex] (blockb1) --node[inner sep=1,fill=white] {Yes} (blockb2);

\draw[-latex] (blockc1) --node[inner sep=1,fill=white] {Yes}(blockc2);

\draw[-latex] (blockb2) -| node[pos=0.55,yshift=0cm,inner sep=1,fill=white] {Yes} (blockb31);
\draw[-latex] (blockb2) -|node[pos=0.55,yshift=-0.5cm,inner sep=1,fill=white]{No} (blockb32);
 
\draw[-latex] (blockb31) -| node[pos=0.55,yshift=-1cm,inner sep=1,fill=white] {No} (blockb41);
\draw[-latex] (blockb31) -|node[pos=0.55,yshift=-0.75cm,inner sep=1,fill=white]{Yes} (blockb42);

\draw[-latex] (blockb32) -| node[pos=0.55,yshift=-0.5cm,inner sep=1,fill=white] {Yes} (blockb43);
\draw[-latex] (blockb32) -|node[pos=0.55,yshift=-0.55cm,inner sep=1,fill=white]{No} (blockb44);

\draw[-latex] (blockb41) -| node[pos=0.55,yshift=-0.5cm,inner sep=1,fill=white] {No} (blockb51);
\draw[-latex] (blockb41) -|node[pos=0.55,yshift=-0.65cm,inner sep=1,fill=white]{Yes} (blockb52);

\draw[-latex] (blockb42) -| node[pos=0.55,yshift=-0.5cm,inner sep=1,fill=white] {Yes} (blockb53);
\draw[-latex] (blockb42) -|node[pos=0.55,yshift=-0.5cm,inner sep=1,fill=white]{No} (blockb54);

\draw[-latex] (blockb52) -| node[pos=0.55,yshift=-0.5cm,inner sep=1,fill=white] {Yes} (blockb61);
\draw[-latex] (blockb52) --node[pos=0.55,yshift=0cm,inner sep=1,fill=white]{No} (blockb62);

\draw[-latex] (blockc2) -| node[pos=0.55,yshift=-0.5cm,inner sep=1,fill=white] {No} (blockc31);
\draw[-latex] (blockc2) -|node[pos=0.55,yshift=-0.5cm,inner sep=1,fill=white]{Yes} (blockc32);

\draw[-latex] (blockc31) -| node[pos=0.55,yshift=-0.5cm,inner sep=1,fill=white] {No} (blockc41);
\draw[-latex] (blockc31) -|node[pos=0.55,yshift=-0.5cm,inner sep=1,fill=white]{Yes} (blockc42);

\draw[-latex] (blockc41) -| node[pos=0.55,yshift=-0.5cm,inner sep=1,fill=white] { Yes} (blockc51);
\draw[-latex] (blockc41) -|node[pos=0.55,yshift=-0.5cm,inner sep=1,fill=white]{No} (blockc52);

\draw[-latex] (blockc51) -| node[pos=0.55,yshift=-0.5cm,inner sep=1,fill=white] { Yes} (blockc61);
\draw[-latex] (blockc51) -|node[pos=0.55,yshift=-0.5cm,inner sep=1,fill=white]{No} (blockc62);

\draw[-latex] (blockc42) -| node[pos=0.55,yshift=-1cm,inner sep=1,fill=white] {Yes} (blockc53);
\draw[-latex] (blockc42) -|node[pos=0.55,yshift=-0.5cm,inner sep=1,fill=white]{No } (blockc54);

\draw[-latex] (blockc53) -| node[pos=0.55,yshift=-1cm,inner sep=1,fill=white] {Yes} (blockc63);
\draw[-latex] (blockc53) -|node[pos=0.55,yshift=-0.5cm,inner sep=1,fill=white]{No } (blockc64);
}
\end{tikzpicture}
 \caption{Decision tree for the utility of side entanglement in qubit channel discrimination. Rectangles stand for tests and ellipses for decisions. Symbols refer to Eqs.\eqref{al2}-\eqref{g22} 
 with $\gamma_m$, $\gamma_M$ and $P$ given in \eqref{gmM}, \eqref{Pdef}.  } 
\end{figure}
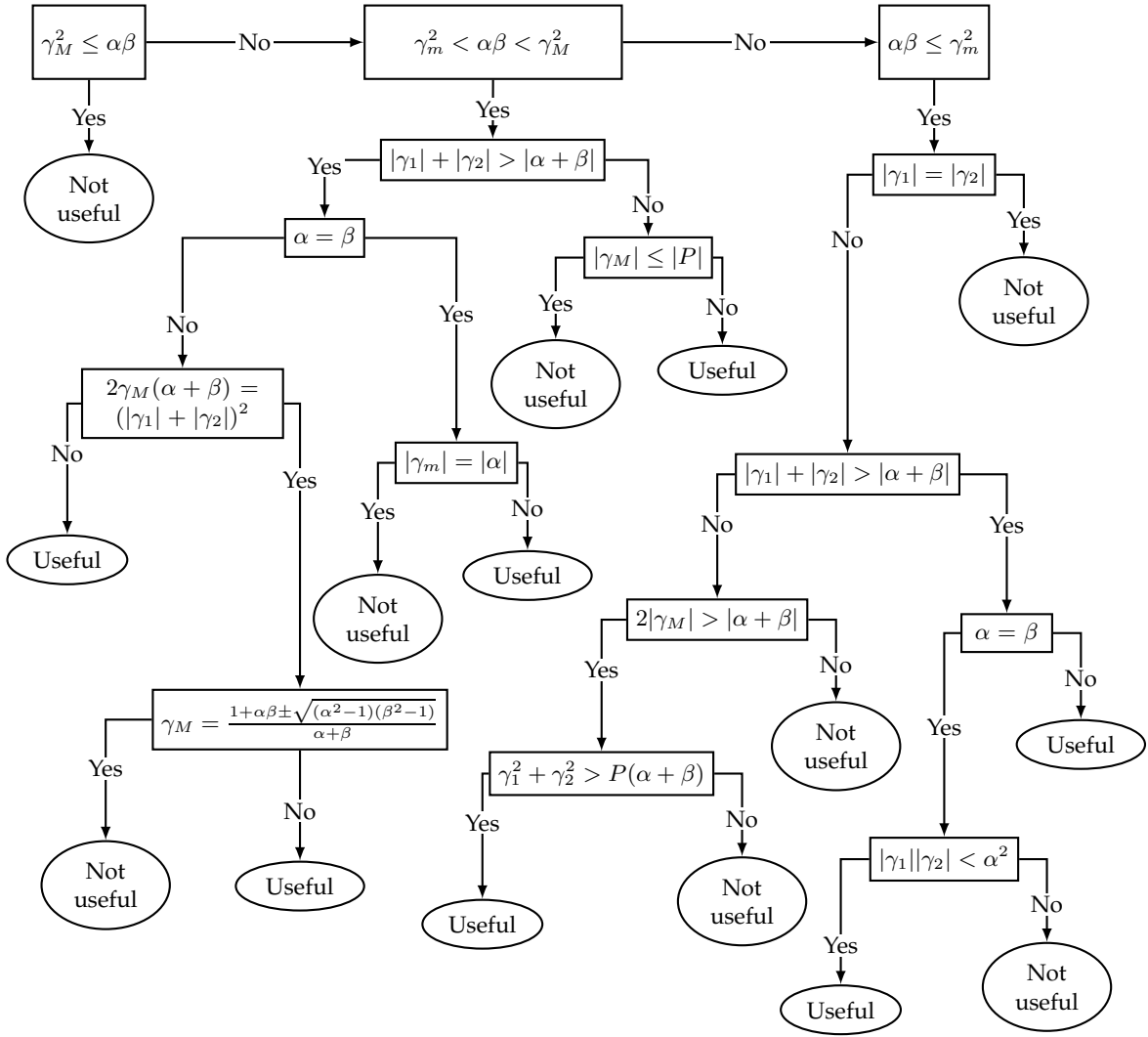


\subsection*{Acknowledgments}
S.M. and M.R. acknowledge funding from the European Union’s
Horizon 2020 research and innovation programme, 
under grant agreement QUARTET No 862644.




\begin{thebibliography}{99}
\bibitem{F01} Analogously entanglement can help in distinguishing between quantum measurements \cite{chandan}.
\bibitem{F02} A single point has to be interpreted as a pair of degenerate points in some special cases.
\bibitem{F03} Actually $\omega_1$ and $\omega_2$ can also be two parameters' vectors.
\end{thebibliography}
\end{document}